\documentclass[10pt,journal,cspaper,compsoc]{IEEEtran}
\usepackage{fixltx2e}
\usepackage{amsmath,amsfonts,amsthm,amssymb}
\usepackage{graphicx}
\usepackage{cite}
\usepackage{algorithm}
\usepackage{algorithmic}
\usepackage{url}
\usepackage[normalem]{ulem}

\usepackage[pdftex]{color}
\usepackage{varioref}
\labelformat{equation}{\textup{(#1)}}
\usepackage[sort&compress,colon,square,numbers]{natbib}

\providecommand{\tk}[1]{\textcolor{black}{#1}}

\providecommand{\ve}[1]{\boldsymbol{#1}}

\newcommand{\argmax}{\operatornamewithlimits{argmax}}
\newcommand{\argmin}{\operatornamewithlimits{argmin}}

\newcommand{\GG}{\mathbb{G}}         
\newcommand{\EE}{\mathbb{E}}           
\newcommand{\II}{\mathbb{I}}           

\newcommand{\thma}{\begin{thm}}
\newcommand{\thmb}{\end{thm}}

\newtheorem{defi}{Definition}
\newtheorem{thm}{Theorem}

\newcommand{\enuma}{\begin{enumerate}}
\newcommand{\enumb}{\end{enumerate}}
\newcommand{\ena}{\begin{enumerate}}
\newcommand{\enb}{\end{enumerate}}

\newcommand{\itema}{\begin{itemize}}
\newcommand{\itemb}{\end{itemize}}
\newcommand{\ita}{\begin{itemize}}
\newcommand{\itb}{\end{itemize}}

\newcommand{\alignb}{\end{align}}

\newcommand{\proofa}{\begin{proof}}
\newcommand{\proofb}{\end{proof}}

\newcommand{\bla}{\begin{block}}
\newcommand{\blb}{\end{block}}

\newcommand{\seqb}{\end{equation*}}

\newcommand{\bth}{\ve{\theta}}
\newcommand{\bhth}{\wh{\ve{\theta}}}

\newcommand{\bTh}{\ve{\Theta}}

\providecommand{\mc}[1]{\mathcal{#1}}
\providecommand{\mb}[1]{\boldsymbol{#1}}

\providecommand{\mv}[1]{\vec{#1}}

\providecommand{\wh}[1]{\widehat{#1}}

\providecommand{\mhc}[1]{\widehat{\mathcal{#1}}}
\providecommand{\mhb}[1]{\hat{\boldsymbol{#1}}}

\newcommand{\Real}{\mathbb{R}}

\newcommand{\conv}{\rightarrow}

\newcommand{\defa}{\begin{defi}}
\newcommand{\defb}{\end{defi}}
\newcommand{\defeq}{\overset{\triangle}{=}}

\newtheorem{Lem}{Lemma}[section]

\newcommand{\from}{{\ensuremath{\colon}}}           
\newcommand{\comment}[1]{}

\begin{document}

\title{Graph Classification using Signal-Subgraphs: Applications in Statistical Connectomics}

\author{Joshua T.~Vogelstein,
		William~R.~Gray,
        R.~Jacob~Vogelstein,
        and~Carey~E.~Priebe
\IEEEcompsocitemizethanks{\IEEEcompsocthanksitem J.T. Vogelstein and C.E. Priebe are with the Department
of Applied Mathematics and Statistics, Johns Hopkins University, Baltimore, MD 21218. 
E-mail: \{joshuav,cep\}@jhu.edu
\IEEEcompsocthanksitem W.R. Gray and R.J. Vogelstein are with the Johns Hopkins University Applied Physics Laboratory, Laurel, MD, 20723.}
\thanks{}}
 
\markboth{Signal Subgraph Classifier}%
{Graph Classification}

\IEEEcompsoctitleabstractindextext{%
\begin{abstract}
This manuscript considers the following ``graph classification'' question: given a collection of graphs and associated classes, how can one predict the class of a newly observed graph?  To address this question we propose a statistical model for graph/class pairs.  This model naturally leads to a set of estimators to identify the class-conditional signal, or ``signal-subgraph,'' defined as the collection of edges that are probabilistically different between the classes. The estimators admit classifiers which are asymptotically optimal and efficient, but differ by their assumption about the ``coherency'' of the signal-subgraph (coherency is the extent to which the signal-edges ``stick together'' around a common subset of vertices). Via simulation, the best estimator is shown to be not just a function of the coherency of the model, but also the number of training samples.  These estimators are employed to address a contemporary neuroscience question: can we classify ``connectomes'' (brain-graphs) according to sex?  The answer is yes, and significantly better than all benchmark algorithms considered.  Synthetic data analysis demonstrates that even when the model is correct, given the relatively small number of training samples, the estimated signal-subgraph should be taken with a grain of salt.  We conclude by discussing several possible extensions.
\end{abstract}

}

\maketitle
\IEEEdisplaynotcompsoctitleabstractindextext
\IEEEpeerreviewmaketitle

\section{Introduction} \label{sec:intro}

\IEEEPARstart{G}{raphs} are emerging as a prevalent form of data representation in fields ranging from optical character recognition and chemistry \cite{Bunke2011} to neuroscience \cite{Hagmann2010}.  While statistical inference techniques for vector-valued data are widespread, statistical tools for the analysis of graph-valued data are relatively rare \cite{Bunke2011}. In this work we consider the task of \emph{\tk{labeled} graph classification}: given a collection of labeled graphs and their corresponding class\tk{es}\comment{labels}, can we accurately infer the class \comment{label} for a new graph?  \tk{Note that we assume throughout that each vertex has a unique label, and that all graphs have the same number of vertices with the same vertex labels.}

We propose and analyze a joint graph/class model---sufficiently simple to characterize its asymptotic properties, and sufficiently rich to afford useful empirical applications.  This model admits a class-conditional signal encoded in a subset of edges, the \emph{signal-subgraph}. Finding the signal-subgraph amounts to providing an understanding of the differences between the two graph classes.  Moreover,
\tk{borrowing a term from the compressive sensing literature \cite{Donoho2006, Candes2008}}, 
 we are interested in learning to what extent this signal is \emph{coherent}; that is, to what extent are the signal-subgraph edges incident to a relatively small set of vertices. \tk{In other words, if the signal is sparse in the edges, then the signal-subgraph is incoherent, if it is also sparse in the vertices, then the signal-subgraph is coherent (we formally define these notions below).} \comment{If the signal-subgraph is strongly coherent, this suggests that the signal is carried by a few important vertices in the graph; otherwise, the signal is more widely distributed across the graph, with no particularly special vertices. }

This graph-model based approach is qualitatively different from most previous approaches which utilize only \tk{unique} vertex labels or graph structure.  In the former case, simply representing the adjacency matrix with a vector and applying standard machine learning techniques ignores graph structure (for instance, it is not clear how to implement a coherent signal-subgraph estimator in this representation).  In the latter case, computing a set of graph invariants (such as clustering coefficient), and then classifying using only these invariants ignores vertex labels \tk{\cite{Kudo2005,Ketkar2009,Bunke2011}}.  \comment{Neither of these approaches use both vertex labels and graph structure.}

\tk{While some of the above approaches consider attributed vertices or edges, we are unable to find any that utilize both \emph{unique} vertex labels and graph structure. The field of \emph{connectomics} (the study of brain-graphs), however, is ripe with many examples of brain-graphs with vertex labels.  In invertebrate brain-graphs, for example, often each neuron is named, such that one can compare neurons across individuals of the same species \cite{North2007}.  In vertebrate neurobiology, while neurons are rarely named, ``neuron types''  \cite{Shepherd2007} and neuroanatomical regions \cite{Nolte2002} are named.  Moreover, a widely held view is that many psychiatric issues are fundamentally ``connectopathies'' \cite{LichtmanSanes08, Bassett2009}.  For prognostic and diagnostic purposes, merely being able to differentiate groups of brain-graphs from one another is sufficient.  However, for treatment, it is desirable to know which vertices and/or edges are malfunctioning, such that therapy can be targeted to those locations. This is the motivating application for our work.}

We demonstrate via theory, simulation, analysis of a neurobiological data set (magnetic resonance based connectome \tk{sex classification}), and synthetic data analysis, that utilizing graph structure can significantly enhance classification accuracy.  However, the best approach for any particular data set is not just a function of the model, but also the amount of data.  Moreover, even when the model is true, given a relatively small sample size, the estimated signal-subgraph will often overlap with the truth, but not fully capture it.  Nonetheless, the classifiers described below still significantly outperform the benchmarks.

\section{Methods} 
\label{sec:methods}


\subsection{Setting}

Let $\GG: \Omega \to  \mc{G}$ be a graph-valued random variable with samples $G_i$.  Each graph $G=(\mc{V},E)$ is defined by a set of $V$ vertices, $\mc{V}=\{v_i\}_{i \in [V]}$, where $[V]=\{1,\ldots, V\}$, and a set of  edges between pairs of vertices $E \subseteq V \times V$. Let $A: \Omega \to  \mc{A}$ be an adjacency matrix-valued random variable taking values $a \in \mc{A} \subseteq \Real^{V \times V}$, identifying which vertices share an edge. Let $Y:\Omega \to  \mc{Y}$ be a discrete-valued random variable with samples $y_i$.  Assume the existence of a collection of $n$ exchangeable samples of graphs and their corresponding classes from some true but unknown joint distribution: $\{(\GG_i,Y_i)\}_{i \in [n]} \overset{exch.}{\sim} F_{\GG,Y}$. Our aim (exploitation task) is to build a graph classifier that can take a new graph, $\GG$, and correctly estimate its class, $y$, assuming that they are jointly sampled from \comment{the same} \tk{some} distribution, $F_{\GG,Y}$.  Moreover, we are interested solely in graph classifiers that are \emph{interpretable} with respect to the vertices and edges of the graph. In other words, nonlinear manifold learning, feature extraction, and related approaches are unacceptable.  

\tk{We adopt the common practice of identifying graphs with their adjacency matrices.  We note, however, that operations available on the latter (addition, multiplication) are not intrinsic to the former.}

\subsection{Model} 
\label{sub:model}

\comment{A model defines the set of distributions under consideration.  In the graph classification domain, we} Consider the model, $\mc{F}_{\GG,Y}$, which includes all joint distributions over graphs and classes under consideration: $\mc{F}_{\GG,Y}=\{F_{\GG, Y}(\cdot; \bth) : \bth \in \bTh\}$, where $\bth \in \bTh$ indexes the distributions.  \comment{Two standard approaches for tackling a classification problem are (i) the \emph{generative} approach and (ii) the \emph{discriminative} approach.  In a generative strategy, one decomposes the joint distribution into a product of a likelihood term and a prior term:  $F_{\GG,Y}=F_{\GG | Y}F_Y$.  In a discriminative strategy, one decomposes the joint distribution into a posterior term and a marginal term: $F_{\GG,Y}=F_{Y | \GG}F_{\GG}$.}  We proceed via a hybrid generative-discriminative approach \tk{\cite{Lasserre2006}} whereby we describe a generative model and place constraints on the discriminant boundary.

First,  assume that each graph has the same set of uniquely labeled vertices, so that all the variability in the graphs is in the adjacency matrix, which implies that $F_{\GG,Y}=F_{A,Y}$. Second, assume edges are independent; that is, $F_{A,Y}=\prod_{u,v \in \mc{E}} F_{A_{uv},Y}$, where 
$\mc{E} \tk{\subseteq} V \times V$ is the set of all possible edges.  Now, consider the generative decomposition \tk{$F_{A,Y}=F_{A|Y} F_Y$}, and let $F_{uv|y}=F_{A_{uv} | Y=y}$ and $\pi_y=F_{Y=y}$.  Third, assume the existence of a class-conditional difference; that is, $F_{uv|0} \neq F_{uv|1}$ for some $(u,v) \in \mc{E}$, and denote the edges satisfying this condition  the \emph{signal-subgraph}, $\mc{S}=\{(u,v) \in \mc{E}: F_{uv|0} \neq F_{uv|1}\}$.  Fourth, \tk{although the following theory and algorithms are valid for both directed and undirected graphs,} for concreteness, assume that the graphs are \emph{simple} graphs; that is, undirected, with binary edges, and lacking (self-) loops \tk{(so $\mc{E}=\binom{V}{2}$)}.  Thus, the likelihood of an edge between vertex $u$ and $v$ is given by a Bernoulli random variable with a scalar probability parameter:  $F_{uv|y}(A_{uv})=\text{Bern}(A_{uv}; p_{uv|y})$. Together, these four assumptions imply the following model: 
\begin{equation}
\mc{F}_{\GG,Y}=\{F_{A, Y}(a,y; \bth) \quad \forall a\in\mc{A},y\in\mc{Y}: \bth \in \bTh\},
\end{equation} 
where
\begin{multline} \label{eq:model}
F_{A,Y}(a,y; \theta) =  \prod_{uv \in \mc{S}} \text{Bern}(a_{uv}; p_{uv|y})  \pi_y 
\\ \times \prod_{uv \in \mc{E} \backslash \mc{S}} \text{Bern}(a_{uv}; p_{uv}),
\end{multline}
and $\bth=\{\mb{p},\mb{\pi},\mc{S}\}$. The likelihood parameter is constrained such that each element must be between zero and one: $\mb{p}\in (0,1)^{\binom{V}{2} \times |\mc{Y}|}$.  The prior parameter, $\mb{\pi}=(\pi_1, \ldots, \pi_{|\mc{Y}|})$, must have elements greater than or equal to zero and sum to one: $\pi_y \geq 0,$ $\sum_y \pi_y=1$.  The signal-subgraph parameter is a non-empty subset of the set of possible edges, $\mc{S} \subseteq \mc{E}$ and $\mc{S} \neq \emptyset$.

We consider up to two additional constraints on $\mc{S}$.  First, the size of the signal-subgraph may be constrained such that $|\mc{S}| \leq s$. Second, the minimum number of vertices onto which the collection of edges is incident to is constrained such that $\mc{S}=\{(u,v): u \cup v \in \mc{U}\}$, where $\mc{U}$ is a set of \emph{signal-vertices} with $|\mc{U}|\leq m$. 
Edges in the signal-subgraph are called \emph{signal-edges}. \tk{Note that given a collection of signal-edges, the signal-vertex set may not be unique.  While it may be natural to treat $\mc{S}$ as a prior, we treat it as a parameter of the model; the constraints, $s$ and $m$, are considered \emph{hyper-parameters}.}

Note that given a specification of the class-conditional likelihood of each edge and class-prior, one completely defines a joint distribution over graphs and classes; the signal-subgraph is implicit in that parameterization. However, the likelihood parameters for all edges not in the signal-subgraph, $p_{uv|y}=p_{uv} \, \forall \, y \in \mc{Y}, (u,v) \notin \mc{S}$,  are \emph{nuisance} parameters; that is, they contain no class-conditional signal.  When computing a relative posterior class estimate, these nuisance parameters cancel in the ratio.


\subsection{Classifier} 
\label{sub:classifier}

\comment{Formally, we say that} A graph classifier, $h \in \mc{H}$, is any function satisfying $h: \mc{G} \to \mc{Y}$.  We desire the ``best'' possible classifier, $h_*$. To define best, we first choose a loss function, $\ell_h: \mc{G} \times \mc{Y} \to \Real_+$, specifically the $0-1$ loss function:
\begin{align}
\ell_h(G,y) \defeq \II \{h(G) \neq y\}, 
\end{align}
where $\II\{\cdot\}$ is the indicator function, equaling one whenever its argument is true, and zero otherwise.  Further, let risk, $R: \mc{F} \times \mc{H} \to \Real_+$ be the expected loss under the true distribution:
\begin{align}
R(F,h) \defeq \EE_F[\ell_h(\GG,Y)].
\end{align}
The Bayes optimal (best) classifier for a given distribution $F$ minimizes risk.
It can be shown that the classifier that maximizes the class-conditional posterior $F_{Y | \GG}$ is optimal \cite{Bickel2000}:
\begin{align} \label{eq:map}
h_* &= \argmin_{h \in \mc{H}} \EE_F[\ell_h (\GG,Y)] 
\nonumber \\ &= \argmax_{y \in \mc{Y}} F_{\GG|Y=y} F_{Y=y}.
\end{align}
Given the proposed model, Eq. \ref{eq:map} can be further factorized using the above four assumptions:
\begin{align}
h_*(G)
&= \argmax_{y \in \mc{Y}} \prod_{u,v \in \mc{S}} \text{Bern}(A_{uv}; p_{uv|y}) \pi_y.
\end{align}
Unfortunately Bayes optimal classifiers are typically unavailable. In such settings, it is therefore desirable to induce a classifier estimate from a set of \emph{training data}. Formally, let $\mc{T}_n= \{(\GG_i,Y_i)\}_{i \in [n]}$ denote the training corpus, where each graph-class pair is sampled exchangeably from the true but unknown distribution: $(\GG_i,Y_i) \overset{exch.}{\sim} F_{\GG, Y}$.  Given such a training corpus and an unclassified graph $G$, an induced classifier predicts the true (but unknown) class of $G$, $\wh{h}\from \mc{G} \times (\mc{G} \times \mc{Y})^n  \to \mc{Y}$.  When a model $\mc{F}_{\GG,Y}$ is specified, a beloved approach is to use a  \emph{Bayes plugin classifier}. Due to the above simplifying assumptions, the Bayes plugin classifier for this model is defined as follows.  First, estimate the  model parameters $\bth=\{\mc{S}, \mb{p}, \mb{\pi}\}$. Second, plug those estimates into the above equation.  The result is a Bayes plugin graph classifier:
\begin{align}
\wh{h}(G; \mc{T}_n) \defeq  \argmax_{y \in \mc{Y}} \prod_{u,v \in \mhc{S}}
\wh{p}_{uv|y}^{a_{uv}}(1-\wh{p}_{uv|y})^{(1-a_{uv})} \wh{\pi}_y,
\end{align}
where the Bernoulli probability is explicit. To implement such a classifier estimate, we specify estimators for $\mc{S}$, $\mb{\pi}$ and $\mb{p}$.


\subsection{Estimators} 
\label{sub:estimators}

\tk{\subsubsection{Desiderata}}

\comment{In this section we describe estimators for the above model.  An \emph{estimator} is a function that maps from the multiple-sample space to the parameter space: $\bhth_n: \Xi^n \to \bTh$. The output of this function is called the \emph{estimate}.  In the graph classification domain, for example, $\Xi=\mc{G} \times \mc{Y}$.  In a slight abuse of notation, we will also refer to the sequence of estimators, $\bhth_1,\bhth_2, \ldots$, as an estimator.}  We desire a \comment{(}sequence of\comment{)} estimators\tk{, $\bhth_1,\bhth_2, \ldots$,} that satisfy the following desiderata:

\begin{itemize}
	\item \textbf{Consistent}: an estimator is consistent (in some specified sense) if its sequence converges in the limit to the true value: $\lim_{n \conv \infty} \bhth_n = \bth$.  
	\item \textbf{Robust}: an estimator is robust if the resulting estimate is relatively insensitive to small model misspecifications.  Because the space of models is massive (uncountably infinite), it is intractable to consider all misspecifications, so we consider only a few of them, as described below.
	\item \textbf{Quadratic complexity}: computational time complexity should be no more than quadratic in the number of vertices.
	\item \textbf{Interpretable}: we desire that the parameters are interpretable with respect to a subset of vertices and/or edges.
\end{itemize}
In addition to the above theoretical desiderata, we also desire appealing finite sample and empirical performance.

\subsubsection{Signal-Subgraph Estimators} 
\label{ssub:subsubsection_name1}

Na\"{i}vely, one might consider a search over all possible signal-subgraphs by plugging each one in to the classifier and selecting the best performing option.  This strategy is intractable because the number of signal-subgraphs scales super-exponentially with the number of vertices (see Figure \ref{fig:numgraphs}, left panel). Specifically, the number of possible edges in a simple graph with $V$ vertices is $d_V=\binom{V}{2}$, so the number of unique possible signal-subgraphs is $2^{\binom{V}{2}}$.  Searching over all of them is sufficiently computationally taxing 
\tk{as to motivate the search for other alternatives.}
\comment{We therefore consider several alternatives.}


\begin{figure*}[tb!]
	\centering
		\includegraphics[width=1.0\linewidth]{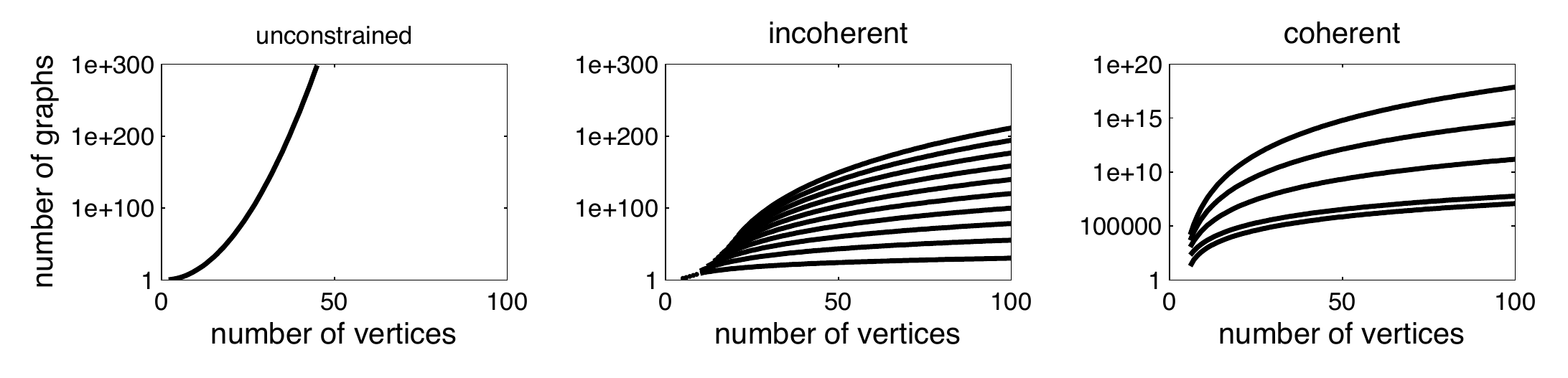}
	\caption{Exhaustive searches for the signal-subgraph, even given severe constraints, are computationally intractable even for small graphs.  The three panels illustrate the 
	the number of unique simple subgraphs as a function of the number of vertices $V$
	for the three different constraint types considered: unconstrained, edge constrained, and both edge and vertex constrained (coherent).  Note the ordinates are all log scale.   On the left is the unconstrained scenario, that is, all possible subgraphs for a given number of vertices.  In the middle panel, each line shows the number of subgraphs with fixed number of signal-edges, $s$, ranging from 10 to 100, incrementing by 10 with each line.  The right panel shows the number of subgraphs for various fixed $s$ and only a single signal-vertex; that is, all edges are incident to one vertex.  
	}
	\label{fig:numgraphs}
\end{figure*}

\begin{figure*}[tb!]
	\centering
		\includegraphics[width=1.0\linewidth]{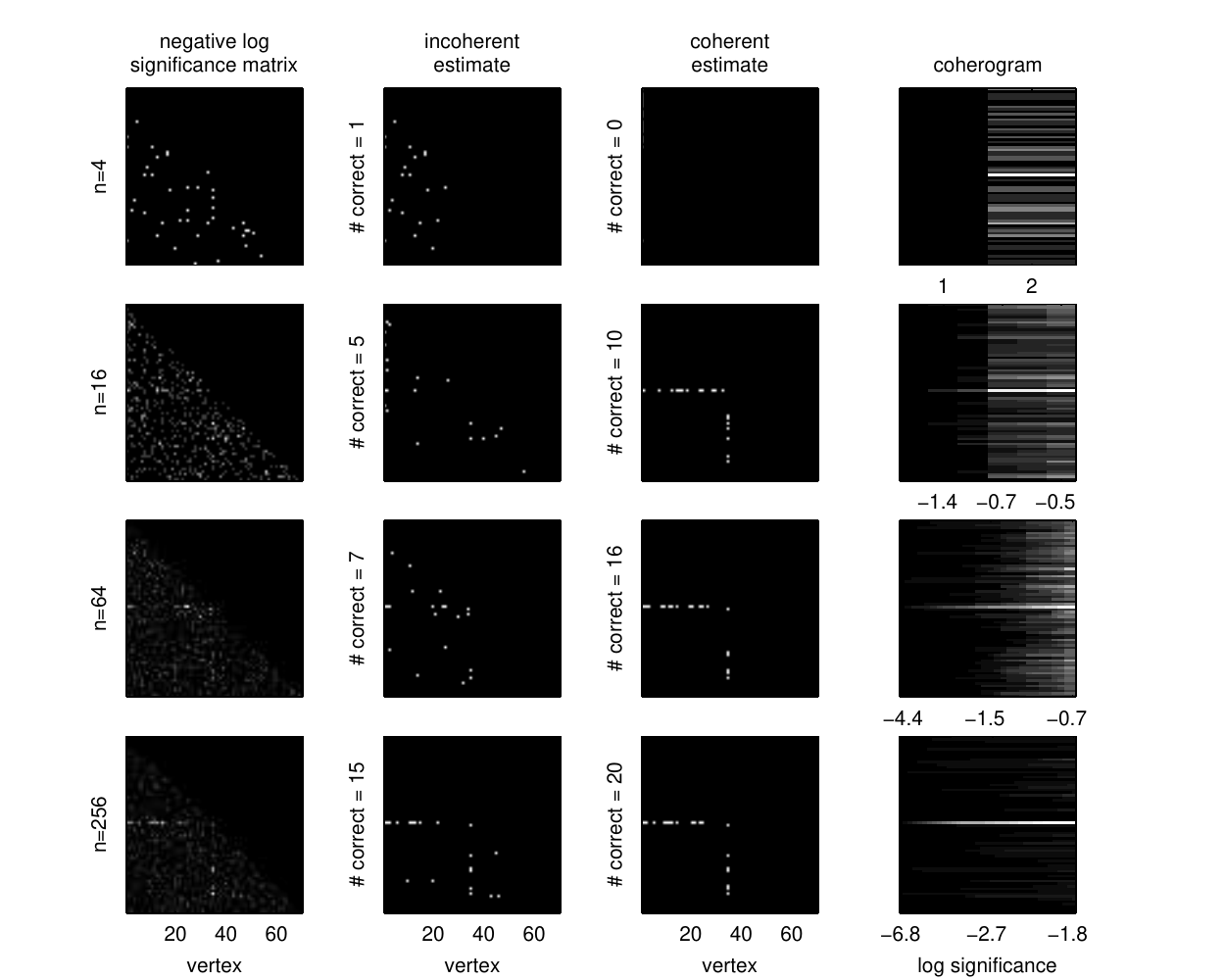}
	\caption{An example of the coherent signal-subgraph estimate's improved accuracy over the incoherent signal-subgraph estimate, for a particular homogeneous two-class model specified by: $\mc{M}_{70}(1,20;0.5,0.1,0.3)$. Each row shows the same columns but for increasing the number of graph/class samples.  The columns show the: (far left) negative log-significant matrix, computed using Fisher's exact test (lighter means more significant; each panel is scaled independent of the others because only relative significance matters here); (middle left) incoherent estimate of the signal-subgraph; (middle right) coherent estimate of the signal-subgraph; (far right) coherogram.  As the number of training samples increases (lower rows), both the incoherent and coherent estimates converge to the truth (the ordinate labels of the middle panels indicate the number of edges correctly identified).  For these examples, the coherent estimator tends to find more true edges.  The coherogram visually depicts the coherency of the signal; it is also converging to the truth---the signal-subgraph here contains a single signal-vertex.}
	\label{fig:4x4}
\end{figure*}

Before proceeding, recall that each edge is independent; thus, one can evaluate each edge separately (although \tk{treating edges independently is} not necessarily advisable, consider the Stein estimator \cite{Stein1956}).  Formally, consider a hypothesis test for each edge.  The \tk{simple} null hypothesis is that the class-conditional edge distributions are the same, so $H_0: F_{uv|0}=F_{uv|1}$.  The composite alternative hypothesis is that they differ, $H_A: F_{uv|0} \neq F_{uv|1}$.  Given such hypothesis tests, one can construct test statistics $T_{uv}^{(n)}: \mc{T}_n \to \Real_+$.  We reject the null in favor of the alternative whenever the value of the test statistic is greater than some critical-value: $T_{uv}^{(n)}(\mc{T}_n)>c$.  We can therefore construct a \emph{significance matrix} $\mb{T} \defeq T_{uv}^{(n)}$, which is the sufficient statistic for the signal-subgraph estimators. 
\tk{Example test statistics include Fisher's and chi-squared, which will be discussed further below.}
Whichever test statistic one uses, the sufficient statistics are captured in a $2 \times |\mc{Y}|$  contingency table, indicating the number of times edge $u,v$ was observed in each class.  For example, the two-class contingency table for each edge is given by:  
%
\begin{table}[h!]
\begin{center}
\begin{tabular}{c||c|c||c}
 & Class 0  & Class 1 & Total \\
\hline\hline
Edge & $n_{uv|0}$ & $n_{uv|1}$ & $n_{uv}$ \\ \hline
No Edge & $n_0-n_{uv|0}$ & $n_1-n_{uv|1}$ & $n-n_{uv}$ \\ \hline \hline
Total & $n_0$ & $n_1$ & $n$. 
\end{tabular}
\end{center}
\label{tab:contingency}
\end{table}%

\tk{For simplicity, we will assume that $|\mc{Y}|=2$ for the remainder, though the general case is  relatively straightforward.}

\paragraph{\emph{Incoherent signal-subgraph Estimators}} 
\label{par:paragraph_name}

Assume the size of the signal-subgraph, $|\mc{E}|=s$, is known.  The number of subgraphs with $s$ edges on $V$ vertices is given by $\binom{d_V}{s}$; also super-exponential (see Figure \ref{fig:numgraphs}, middle panel). Thus, searching them all is currently computationally intractable.  When $s$ is given, under the  independent edge assumption, one can choose the critical value \emph{a posteriori} to ensure that only $s$ edges are rejected \tk{under the null (that is, have significant class-conditional differences)}:
\begin{align}
	&\text{minimize } c \nonumber \\ &\text{subject to } \sum_{(u,v) \in \mc{E}} \II \{T_{uv}^{(n)} < c\} \geq s.
\end{align}
Therefore, an estimate of the signal-subgraph is the collection of $s$ edges with minimal test statistics.  Let $T_{(1)} < T_{(2)} < \cdots < T_{(d_V)}$ indicate the \emph{ordered} test statistics (dropping the superscript indicating the number of samples for brevity).  Then, the \emph{incoherent signal-subgraph estimator} is given by $\mhc{S}_n(s)=\{e_{(1)}, \ldots, e_{(s)}\}$, where $e_{(u)}$ indicates the $u^{th}$ edge ordered by significance of its test statistic, $T_{(u)}$.  

\tk{Note that the number of distinct test-statistic values is typically much smaller than the number of possible settings of $s$; specifically, the number of unique test statistic values will be $t \leq \min(|\mc{E}|,(n_0+1) (n_1+1))$.  In practice, $t$ is often be far less than either of the upper bounds, because not every edge has a unique contingency table. In such scenarios, certain settings of the hyper-parameters will lead to ``ties'', that is, edges that are equally valid under the assumptions.  In such settings, we simply randomly choose edges satisfying the criterion.}

\tk{Pseudocode for implementing the incoherent signal-subgraph estimator is provided in Algorithm \ref{alg:inc}, and MATLAB code is available from \url{http://jovo.me}.}

\begin{algorithm}
\caption{Pseudocode for estimating incoherent signal-subgraph.}
\label{alg:inc}
\begin{algorithmic}[1]
\REQUIRE $\mc{T}_n$ and $s$
\ENSURE $\mhc{S}_n(s)$
\STATE Compute test statistics $T_{uv}^{(n)}$ for all $(u,v) \in \mc{E}$
\STATE Sort each edge according to its test-statistic rank, $T_{(1)} \leq T_{(2)} \leq \cdots \leq T_{(d_V)}$
\STATE Let $\mhc{S}_n(s)=\{e_{(1)}, \ldots, e_{(s)}\}$, arbitrarily breaking ties as necessary.
\end{algorithmic}
\end{algorithm}

\paragraph{\emph{Coherent Signal-Subgraph Estimators}}

In addition to the size of the signal-subgraph, also assume that each of the edges in the signal-subgraph are incident to one of $m$ special vertices called \emph{signal-vertices}. While this assumption further constrains the candidate sets of edges, the number of feasible sets still scales super exponentially (see Figure \ref{fig:numgraphs}, right panel).  Therefore, we again take a greedy approach.  

First, compute the significance of each edge, as above, yielding ordered test statistics. \tk{Second, rank edges by significance with respect to each vertex,  $e_{k,(1)} \leq e_{k,(2)} \leq \ldots \leq e_{k,(n-1)}$ for all $k \in \mc{V}$.  Third, initialize the critical value at zero, $c=0$. Fourth, assign each vertex a score equal to the number of edges incident to that vertex more significant than the critical value, $w_{v;c}=\sum_{u \in [V]} \II\{T_{v,u} > c\}$.  
Fifth, sort the vertex significance scores, $w_{(1);c} \geq w_{(2); c} \geq \cdots \geq w_{(V);c}$.
Sixth, check if there exists $m$ vertices whose scores sum to greater than or equal the size of the signal-subgraph, $s$.  That is, check whether the following optimization problem is satisfied:
\begin{align} \label{eq:coh}
	&\text{minimize } c \nonumber \\
	&\text{subject to } \sum_{v \in [m]} w_{(v);c}\geq s.
\end{align}
If so, call the collection of $s$ most significant edges from within that subset the \emph{coherent signal-subgraph estimate}, $\mhc{S}_n(s,m)$. 
If not, increase $c$ and go back to step four. 
As above, we break ties arbitrarily.
Pseudocode for implementing the coherent signal-subgraph estimator is provided in Algorithm \ref{alg:coh}, and MATLAB code is available from \url{http://jovo.me}.}

\begin{algorithm}
\caption{Pseudocode for estimating coherent signal-subgraph.}
\label{alg:coh}
\begin{algorithmic}[1]
\REQUIRE $\mc{T}_n$ and $(s,m)$
\ENSURE $\mhc{S}_n(s,m)$
\STATE Compute test statistics $T_{uv}^{(n)}$ for all $(u,v) \in \mc{E}$
\STATE Sort each edge according to its vertex-conditional test-statistic rank, $T_{(1),k} \leq T_{(2),k} \leq \cdots \leq T_{(d_V),k}$ for all $k \in \mc{V}$
\STATE Let $c=0$
\STATE Let $w_{v;c}=\sum_{u \in \mc{V}} \II\{T_{v,u}>c\}$ for all $v \in \mc{V}$ 
\STATE Let $w_c=\sum_{v \in [m]} w_{(v);c}$
\WHILE{$w_c < s$}
\STATE Let $c \leftarrow c+ 1$
\STATE Update $w_c$ 
\ENDWHILE
\STATE Let $\mhc{S}_n(s,m)$ be the collection of $s$ edges from amongst those that satisfy Eq. \ref{eq:coh} for the final value of $c$, arbitrarily breaking ties as necessary.
\end{algorithmic}
\end{algorithm}

\paragraph{\emph{Coherograms}}

In the process of estimating the incoherent signal-subgraph, one builds a ``coherogram''.  Each column of the coherogram corresponds to a different critical value $c$, and each row corresponds to a different vertex $v$.  The $(c,v)^{th}$ element of the coherogram $w_{v;c}$ is the number of edges incident to vertex $v$ with test statistic larger than $c$.  Thus, the coherogram gives a visual depiction of the coherence of the signal-subgraph (see Figure \ref{fig:4x4}, right column, for some examples).

\subsubsection{Likelihood Estimators} 
\label{sub:likelihood}

The class-conditional likelihood parameters $p_{uv|y}$ are relatively simple.  In particular, because the graphs are assumed to be simple, $p_{uv|y}$ is just a Bernoulli parameter for each edge in each class.  The maximum likelihood estimator (MLE), which is simply the average value of each edge per class, is a principled choice:
\begin{align}
\wh{p}_{uv|y}^{MLE} = \frac{1}{n_y} \sum_{i | y_i = y} a_{uv}^{(i)},
\end{align}
where $\sum_{i | y_i=y}$ indicates the sum is over all training samples from class y. Unfortunately, the MLE has an undesirable property; in particular, if the data contains no examples of an edge in a particular class, then the MLE will be zero.  If the unclassified graph exhibits that edge, then the estimated probability of it being from that class is zero, which is undesirable. We therefore consider a smoothed estimator:
\begin{align} \label{eq:lik}
\wh{p}_{uv|y} = 
\begin{cases}
\eta_n & \text{if } \max_i a_{uv}^{(i)}=0 \\
1- \eta_n & \text{if } \min_i a_{uv}^{(i)}=1 \\
\wh{p}_{uv|y}^{MLE} & \text{otherwise}
\end{cases}
\end{align}
where we let $\eta_n=1/(10n)$.  



\subsubsection{Prior Estimators} 
\label{sub:prior_estimators}

The priors are the simplest.  The prior probabilities are Bernoulli, and we are concerned only with the case where $|\mc{Y}| \ll n$, so the maximum likelihood estimators suffice:
\begin{align} \label{eq:prior}
\wh{\pi}_y = \frac{n_y}{n},
\end{align}
where $n_y=\sum_{i \in [n]} \II\{y_i = y\}$.


\tk{\subsubsection{Hyper-Parameter Selection}}

\tk{The signal-subgraph estimators require specifying the number of signal-edges $s$, as well as the number of signal-vertices $m$ for the coherent classifier.  In both cases, the number of possible values of finite.  In particular, $s \in [d_V]$ and $m \in [V]$.  Thus, 
to select the best hyper-parameters
we implement cross-validation procedures (see Section \ref{ssub:classifier} for details), iterating over $(s,m) \in \mv{s} \times \mv{m} \subseteq [d_V] \times [V]$. Note that when $m=V$, the coherent signal subgraph estimator reduces to the incoherent signal subgraph estimator. 
For all simulated data, we compare hyper-parameter performance via a training and held-out set.  For the real data application, we decided to use a leave-one-out cross-validation procedure due to the small sample size.}

 %

\tk{\subsubsection{All together}} 
\label{ssub:all_together}

\tk{Putting the above pieces together, Algorithm \ref{alg:3} provides pseudo-code for implementing our signal-subgraph classifiers. MATLAB code is available from the first author's website, \url{http://jovo.me}.}

\begin{algorithm}
\caption{Pseudocode for training signal-subgraph classifiers.}
\label{alg:3}
\begin{algorithmic}[1]
\REQUIRE $\mc{T}_n$ and a set of constraints $(\mv{s},\mv{m})$
\ENSURE $\mhc{S}_n$, $\{\wh{p}_{uv|y}\}_{(u,v) \in \mhc{S}_n}, \{\wh{\pi}_{y}\}_{y \in \{0,1\}}$
\STATE Partition the data for the appropriate cross-validation procedure
\STATE Estimate $p_{uv|y}$ for all $(u,v)$ using Eq. \ref{eq:lik}
\STATE Estimate $\pi_y$ for all $y$ using Eq. \ref{eq:prior}
\FORALL{$(s,m) \in (\mv{s},\mv{m})$}
\STATE Compute $\mhc{S}_n(s,m)$ using Algorithm \ref{alg:inc} or \ref{alg:coh}, as appropriate
\STATE Compute cross-validated error $\wh{L}_{s,m}$ using Eq. \ref{eq:L2}
\ENDFOR
\STATE Let $\mhc{S}_n=\argmin_{(s,m)} \wh{L}_{s,m}$ 
\end{algorithmic}
\end{algorithm}

\subsection{Finite Sample Evaluation Criteria} 
\label{sub:evaluation_criteria}

\subsubsection{Likelihoods and priors} 
\label{ssub:likelihoods_and_priors}

The likelihood and prior estimators will be evaluated with respect to robustness to model misspecifications, finite samples, efficiency, and complexity.


\subsubsection{Classifier} 
\label{ssub:classifier}

We evaluate the classifier's finite sample properties using either held-out or leave-one-out misclassification performance, depending on whether the data is simulated or experimental, respectively.  Formally, given $C$ equally sized subsets of the data, $\{\mc{T}_{1}, \ldots, \mc{T}_{C}\}$, the \emph{cross-validated error} is given by
\begin{align} \label{eq:L2}
	\wh{L}_{\wh{h}(\cdot; \mc{T}_n)} = \frac{1}{C}\sum_{c=1}^C \frac{1}{|\mc{T}_n \backslash \mc{T}_c|}\sum_{G \notin \mc{T}_c} \II\{\wh{h}(G; \mc{T}_{c}) \neq y\}.
\end{align}
Given this definition, let $L_{\mhb{\pi}}$ be the error of the classifier using only the prior estimates, and let $L_*$ be the error for the Bayes optimal classifier.  

To determine whether a classifier is significantly better than chance, we randomly permute the classes of each graph $n_{MC}$ times, and then estimate a na\"ive Bayes classifier using the permuted data, yielding an empirical distribution.  The p-value of a permutation test is the minimum fraction of Monte Carlo permutations that did better than the classifier of interest \cite{Good2010}.  

To determine whether a pair of classifiers are significantly different, we compare the leave-one-out classification results using McNemar's test \cite{McNemar1947}.

\subsubsection{Signal-Subgraph Estimators} 
\label{ssub:signal_subgraph_estimators}


To evaluate absolute performance of the signal-subgraph estimators, we define  ``miss-edge rate'' as the fraction of true edges missed by the signal-subgraph estimator:
\begin{align}
R^x_n = \frac{1}{|\mc{S}|} \sum_{(u,v)\in \mc{S}}\II\{(u,v) \notin \mhc{S}_n\}.
\end{align}
\tk{Note that when $|\mc{S}|$ is fixed, miss-edge rate is a sufficient statistic for all combinations of false/negative positive/negative results.}
Further, we estimate the \emph{relative rate} and \emph{relative efficiency} to evaluate the relative finite sample properties of a pair of consistent estimators. The relative rate is simply $(1-R^{inc}_n)/(1-R^{coh}_n)$.  Relative efficiency is the number of samples required for the coherent estimator to obtain the same rate as the incoherent estimator.

%



%
%
%
%
%

\section{Estimator Properties} 
\label{sec:results}

%

\subsection{Likelihood and Prior Estimators} 
\label{ssub:subsubsection_name4}

\tk{
\begin{Lem}
	$\wh{p}_{uv|y}$ as defined in Eq. \ref{eq:lik}  is an L-estimator.
\end{Lem}
\begin{proof}
Huber defines an L-estimator as an estimator that is a linear combination of (possibly nonlinear functions of) the order statistics of the measurements \cite{Huber1981}.  Indeed, $\wh{p}_{uv|y}$ is a thresholded function of the minimum, maximum, and mean.
\end{proof}
Because L-estimators converge to the MLE, our estimators share all the nice asymptotic properties of the MLE.  }
%
%
Moreover, L-estimators are known to be robust to certain model misspecifications \cite{Huber1981}. The prior estimators are MLE's, and therefore also consistent and efficient.
%
%
%
Both prior and likelihood estimates are trivial to compute, as closed-form analytic solutions are available for both.  

\subsection{Signal-Subgraph Estimators} 
\label{ssub:subsubsection_name5}

A variety of test statistics are available for computing the edge-specific class-conditional signal, $T_{uv}^{(n)}$.  Fisher's exact test computes the probability of obtaining a contingency table equal to, or more extreme than, the table resulting from the null hypothesis: that the two classes have the same probability of sampling an edge.  In other words, Fisher's exact test is the most powerful statistical test assuming independent edges \cite{Rice1995}.  \tk{This leads to the following lemma:
\begin{Lem}
	$\mhc{S}_n(s',m') \to \mc{S}$ as $n \to \infty$ when computing $T_{uv}^{(n)}$ via Fisher's exact test, even when $s$ and $m$ are unknown, as long $s' \geq s$ and $m' \geq m$. 
\end{Lem}
\begin{proof}
	Whenever $p_{uv|0}\neq p_{uv|1}$, the p-value of Fisher's exact test converges to zero; whereas whenever $p_{uv|0}=p_{uv|1}$, the distribution of p-values converges to the uniform distribution on $[0,1]$.  Therefore, Fisher's exact test induces a consistent estimator of the signal-subgraph as $n \conv \infty$, assuming a fixed and finite $V$.  Moreover, as $V \conv \infty$, as long as $V/n \conv 0$, Fisher's exact test remains consistent \cite{Rice1995}. 
\end{proof}
}
While most powerful, computing Fisher's exactly is computationally taxing.  Fortunately, the chi-squared test is asymptotically equivalent to Fisher's test, and therefore shares those convergence properties \cite{Rice1995}.  Even the absolute difference of MLE's, $|\wh{p}_{uv|1}^{MLE}-\wh{p}_{uv|0}^{MLE}|$, which is trivially easy to compute, is asymptotically equivalent to Fisher's \cite{Rice1995} and therefore consistent. 
\comment{The implications of the above convergence properties are that any incoherent signal-subgraph estimated using a consistent test statistic is a consistent signal-subgraph estimator.}  Moreover, the signal-subgraph estimators are robust to a variety of model misspecifications.  Specifically, as long as all the marginal probability of all the edges in the signal-subgraph are different between the two classes, $p_{uv|1}\neq p_{uv|0}$,  and the constraints are upper-bounds on the true values, $s' \geq s$ and $m' \geq m$, 
then any consistent test statistic will yield a consistent signal-subgraph estimator.  \comment{For example, when the signal-subgraph is coherent, even if $m$ is unknown, the incoherent signal-subgraph estimator will converge to the truth.  More generally, even if the independent edge assumption is not satisfied, the incoherent estimator will converge to the truth.}
\comment{Moreover, the coherent signal-subgraph estimator uses the same test statistics.  Thus, it shares the above consistency and robustness properties. } 
Estimating the coherent signal-subgraph is more computationally time consuming than estimating the incoherent signal-subgraph.
What is lost by computational time, however, is typically gained by finite sample efficiency whenever the model  \comment{is coherent} \tk{does not induce too much bias}, as will be shown below.

\subsection{Bayes Plugin Classifier}

\tk{\begin{Lem}
	The Bayes plug-in classifier, using the signal-subgraph, likelihood, and prior estimators described above, is consistent under the model defined by Eq. \ref{eq:model}.
\end{Lem}
}
 
\begin{proof}
A Bayes plugin classifier is a consistent classifier whenever the estimates that are plugged in are consistent \cite{Bickel2000}.  Because the likelihood, prior, and signal-subgraph estimates are all consistent, the Bayes plugin classifier is also consistent.  
\end{proof}

Note that na\"ive Bayes classifiers often exhibit impressive finite sample performance due to their winning the bias-variance trade-off relative to other classifiers \cite{Hand2001}.  In other words, even when edges are highly dependent, because marginal probability estimates are more efficient than joint probability estimates, an independent edge based classifier will often outperform a classifier based on dependencies.


\section{Simulated Experiments} 
\label{sub:subsection_name}

\subsection{Simulation Details} 
\label{sub:simulation_details}

To better assess the finite sample properties of the signal-subgraph estimators, we conduct a number of simulated experiments.  Consider the following \emph{homogeneous} model: each simple graph has $V=70$ vertices.  Class 0 graphs are Erdos-Renyi with probability $p$ for each edge; that is, $f_{uv|0}=p \, \forall \, (u,v) \in \mc{E}$.  Class 1 graphs are a mixture of two Erdos-Renyi models: all edges in the \emph{signal-subgraph} have probability $q$, and all others have probability $p$, so that $f_{uv|1}=q \, \forall (u,v) \in \mc{S}$, and $f_{uv|1}=p \, \forall (u,v) \in \mc{E} \backslash \mc{S}$.  The signal-subgraph is constrained to have $m$ signal-vertices and $s$ signal-edges.  Let the class-prior probabilities be given by $F_{Y=0}=\pi$ and $F_{Y=1}=1-\pi$. Thus, the model is characterized by $F_{\bth}=\mc{M}_V(m,s; \pi,p,q)$, where $V$ is a constant, $m$ and $s$ are hyper-parameters, and $\pi$, $p$ and $q$ are parameters.

\subsection{A Simple Demonstration} 
\label{sub:a_simple_demonstration}
  

To provide some insight with respect to the finite sample performance of the incoherent and coherent signal-subgraph estimators for this model, we run the following simulated experiments, with results depicted in  Figure \ref{fig:4x4}.  In each row we sample from $\mc{M}_{70}(1,20;0.5,0.1,0.3)$  (note that we are actually conditioning on the class-conditional sample size).  Given these $n$ samples, we compute the significance matrix (first column), which contains the sufficient statistics for both estimators.  The incoherent estimator simply chooses the $s$ most significant edges as the signal-subgraph (second column). The coherent estimator jointly estimates both the $m$ signal-vertices and the $s$ signal-edges incident to at least one of those vertices (third column).  The coherogram shows the ``coherency'' of the data (fourth column).    

From this figure, one might notice a few tendencies.  First, both the incoherent and coherent signal-subgraph seem to converge to the true signal-subgraph.  Second, while 
both estimators perform poorly with $n < 16$, 
the coherent estimator converges more quickly than the incoherent estimator.  Third, the coherogram sharpens with additional samples,  showing after only approximately 50 samples that this model is strongly coherent.

\subsection{Quantitative Comparisons} 
\label{sub:quant}

To better characterize the relative performance of the two signal-subgraph estimators, Figure \ref{fig:homo} shows their performance as a function of the number of training samples, $n$, for the $\mc{M}_{70}(1,20;0.5,0.1,0.3)$ model.  The top panel shows the mean and standard error of the missed-edge rate---the fraction of edges incorrectly identified\tk{---averaged over 200 trials}.  For essentially all $n$, the coherent estimator (black \tk{solid} line) performs better than the incoherent estimator (gray \tk{solid} line).  \tk{We also compare the performance of an $\ell_1$-penalized logistic regression classifier  (`lasso' hereafter \cite{Tibs96}).  As expected, the missed edge rate for the lasso (gray dashed line) and the incoherent classifier are about the same. The improvement in signal-edge detection of the coherent signal-subgraph estimator over the incoherent's and lasso's performance} \comment{This} translates directly to improved classification performance (middle panel), where the plugin classifier using the coherent signal-subgraph estimator has a better misclassification rate than \tk{either} the incoherent signal-subgraph classifier \tk{and the lasso} for essentially all $n$. \tk{Note that the incoherent classifier also admits better performance than the lasso.  This is expected---although they are very similar---the incoherent classifier was derived specifically for this joint graph/class model}.  For comparison purposes, the na\"ive Bayes plugin classifier; that is, the classifier that assumes the whole graph is the signal-subgraph, is also shown (\tk{black dashed} \comment{light gray} line).  Note that the performance of all the classifiers is bounded above by $L_{\mhb{\pi}} = 0.5$ and below by $L_* = 0.13$.  Moreover,  $\wh{L}_{nb} > \wh{L}_{lasso} > \wh{L}_{inc} > \wh{L}_{coh}$ for essentially all $n$.  

\tk{An important aspect of any algorithm is compute time, both of training and testing. The signal-subgraph classifiers that we developed are very fast.  Computations essentially amount to computing a test-statistic for all $|\mc{E}|$ edges, then sorting them.  The parameter estimates of the likelihood and prior terms come directly from the same test-statistics used to obtain the significance of each edge.  Thus, obtaining those estimates amounts to essentially computing a mean.  On the other hand, the lasso classifier, which yields \emph{worse} signal detection and misclassification rates than both our classifiers, requires an iterative algorithm for each value on the hyper-parameter path \cite{Tibs96}.  Despite that efficient computational schemes have been developed for searching the whole regularization path \cite{Efron2004a}, such iterative algorithms should be much slower than our classifiers. }

\tk{Indeed, the lower panel of Figure \ref{fig:homo} demonstrates that our MATLAB implementation of the signal-subgraph classifiers are approximately 10 times faster than MATLAB's lasso implementation.  All the results shown in Figure \ref{fig:homo} include errorbars computed from 100 trials, each with 100 held-out samples, demonstrating that for these simulation parameters, the differences are highly significant.  Although the quantitative results may vary for different implementations and different parameter settings, our expectation is that the qualitative results should be consistent.  Because our classifiers have lower risk, better signal identification, and run an order of magnitude faster than the standard, we do not consider lasso in further simulations.}

\begin{figure}[htbp]
	\centering
		\includegraphics[width=1.0\linewidth]{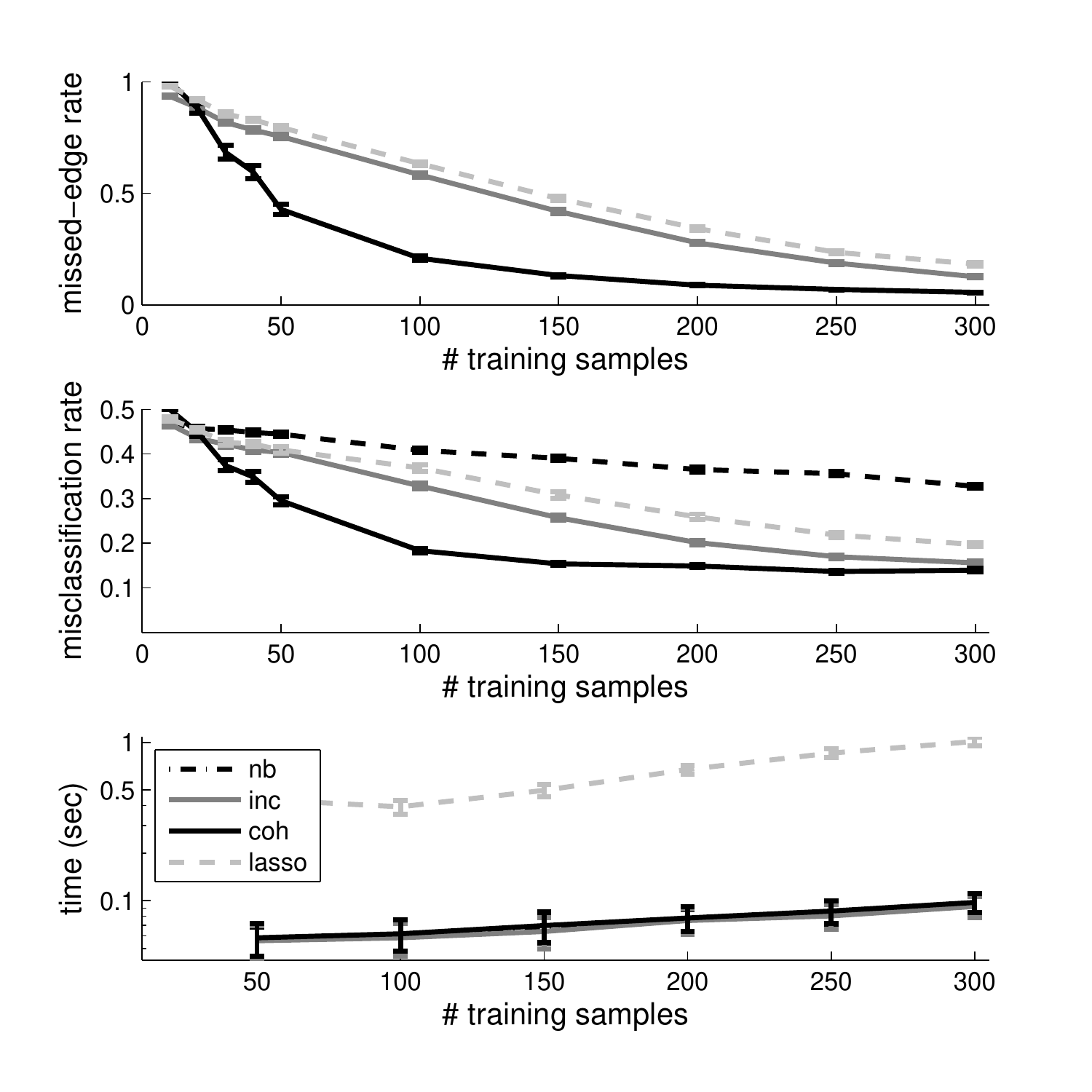}
	\caption{Performance statistics as a function of sample size demonstrate that the coherent signal-subgraph estimator outperforms the incoherent signal-subgraph estimator, in terms of both the signal-subgraph identification and classification, for nearly all $n$, using the same model as in Figure \ref{fig:4x4}: $\mc{M}_{70}(1,20;0.5,0.1,0.3)$.  \tk{Moreover, even the incoherent classifier outperforms the $\ell_1$-penalized logistic regression (lasso) on all our metrics.}	The top panel shows the missed-edge rate for each estimator as a function of the number of training samples, $n$.  The \tk{middle} \comment{bottom} panel shows the corresponding misclassification rate for the estimators, as well as the na\"ive Bayes plugin classifier.  Performance of all estimators improves (nearly) monotonically with $n$ for both criteria. \tk{The bottom panel shows total training and testing time for each classifier.  Clearly, the lasso is about 10 times slower than the others.} Error bars show standard error of the mean here and elsewhere \tk{unless otherwise noted} (averaged over 100 trials; each trial used 100 samples for held-out data). 
	\tk{Error bars on the lower panel show the inter-quartile range.}
	Note that for most values of $n$, we have $L_{\mhb{\pi}} > \wh{L}_{nb} > \tk{\wh{L}_{lasso} >} \wh{L}_{inc} > \wh{L}_{coh} > L_*$. Legend: ``inc'': incoherent; ``coh'': coherent; ``nb'': na\"ive Bayes\tk{, ``lasso'': lasso}.}
	\label{fig:homo}
\end{figure}

The above numerical results suggest that the coherent estimator \tk{achieves better signal-subgraph identification and classification performance than} \comment{outperforms} the incoherent estimator almost always\tk{, despite that the computational time of the coherent classifier is almost identical}.  However, that result is a function of both the model $\mc{M}_V$ (which includes the number of vertices), and the number of training samples $n$ \tk{(there is a bias-variance trade-off here, as always)}.  
Figure \ref{fig:RE} explicitly shows that the relative performance of an estimator for a particular model---$\mc{M}_{30}(1,5;0.5,0.1,0.2)$---changes as a function of the number of samples.  More specifically, for small $n$, the incoherent estimator yields better performance, as indicated by the relative rate and relative efficiency being above one.  However, with more samples, when the signal-subgraph is coherent, the coherent estimator will eventually outperform the incoherent one.  At infinite samples, since both estimators are consistent, they will yield identical results: the truth.  

Thus, to choose which estimator will likely achieve the best performance, knowledge of the model, $\mc{M}_V(m,s;\pi,p,q)$, is insufficient; rather, both the model and the number of samples must be known a priori.  

\begin{figure}[htbp]
	\centering
		\includegraphics[width=0.8\linewidth]{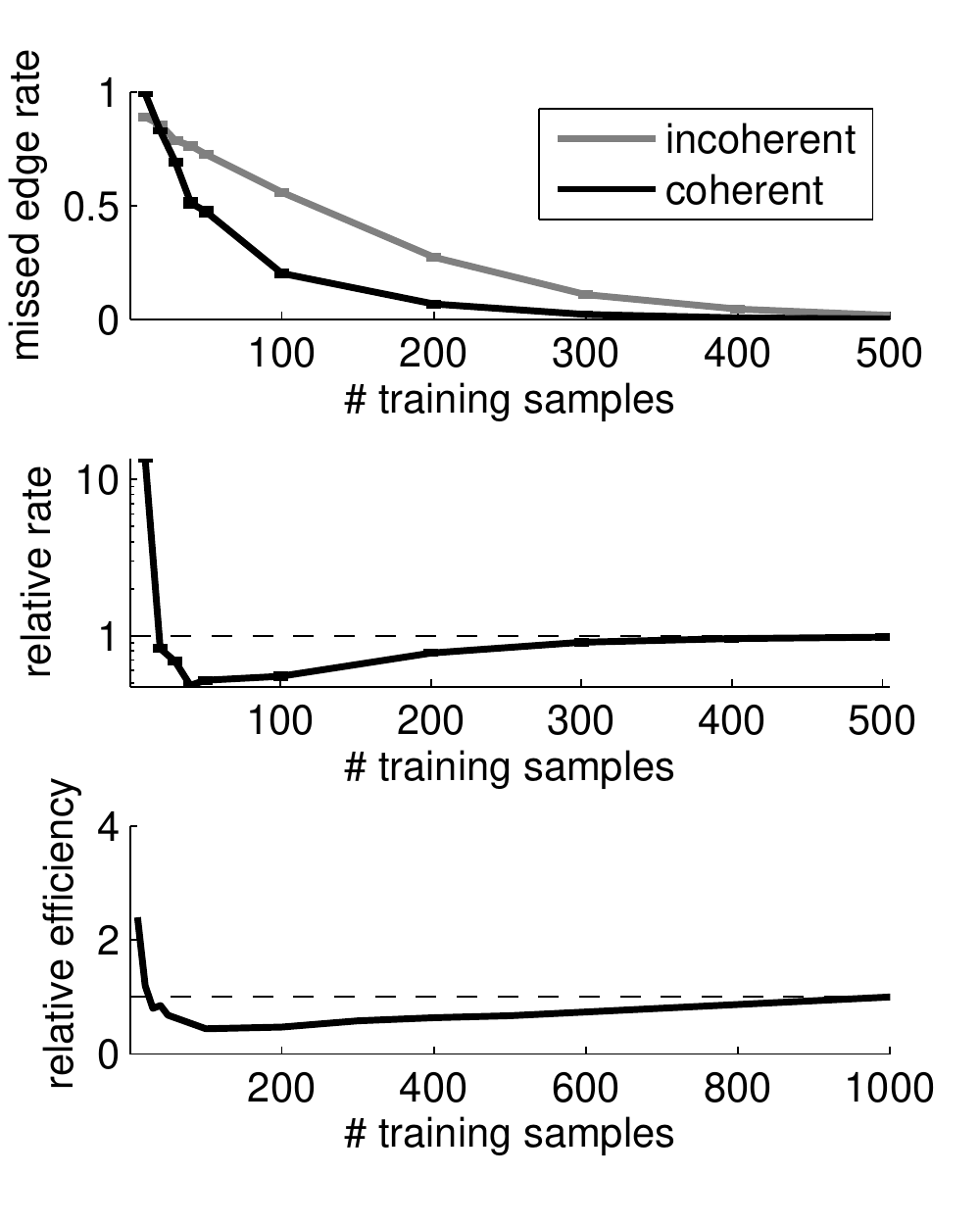}
	\caption{The relative performance of the coherent and incoherent estimators is a function not just of the model, but also the number of training samples.  Specifically, for the same model, $\mc{M}_{70}(1,20;0.5,0.1,0.3)$, we compute the missed-edge rates for both the incoherent estimator (gray line) and the coherent estimator (black line)\tk{, averaged over 200 trials}.  The top panel shows that for small training sample size the incoherent estimator achieves a better (lower) missed-edge rate than the coherent estimator. However, the incoherent estimator's convergence rate is slower, and the coherent estimator catches up and outperforms the incoherent estimator until both eventually converge at the truth.  The middle and bottom panels show the relative rate and efficiency curves for this model. Note that the curves dip below unity, and then converge to unity, as they must, because both estimators are consistent. }
	\label{fig:RE}
\end{figure}

\subsection{Estimating the Hyper-Parameters} 
\label{sub:estimating_the_hyper_parameters}

In the above analyses the hyper-parameters, both the number of signal-edges $s$ and signal-vertices $m$, were known.  In practice while one might have a preliminary guess of the range of these hyper-parameters, the optimal values will usually be unknown.  We can therefore use a cross-validation technique to search over the space of all reasonable combinations of $s$ and $m$, and choose the best performing combination.  Figure \ref{fig:coherent} shows one such simulation depicting several key features.  The top panel shows the misclassification rate \tk{on held-out data} as a function of the log of the assumed size of the signal-subgraph for the incoherent classifier.  Although the true size is $s=20$, the best performing estimate is $\wh{s}_{inc}=23$. This is a relatively standard result in model selection: the best performer will include a few extra dimensions because adding a few uninformative features is less costly than missing a few informative features \cite{Jain2000}.  This intuition is further reified by the U-shape of the misclassification curve on a log scale: including many non-signal-edges is less detrimental than excluding a few signal-edges.

The bottom panel shows the coherent performance by varying both $m$ and $s$, which exhibits a ``banded''  structure, indicating that the performance is relatively robust to small changes in $m$.  \tk{This banding likely results from the fact that the test statistics are identical for many edges, so therefore minor changes in the number of allowable edges is not expected to change performance much.}
The best performing pair achieved $\wh{L}_{coh}=0.13$ (which is equal to the Bayes error) with $\wh{m}_{coh}=1$ and $\wh{s}_{coh}=24$, suggesting that $n$ was sufficiently large to correctly find the true signal-vertex, and further corroborating the ``better safe than sorry'' attitude to selecting the signal-edges. 

\begin{figure}[htbp]
	\centering
		\includegraphics[width=0.8\linewidth]{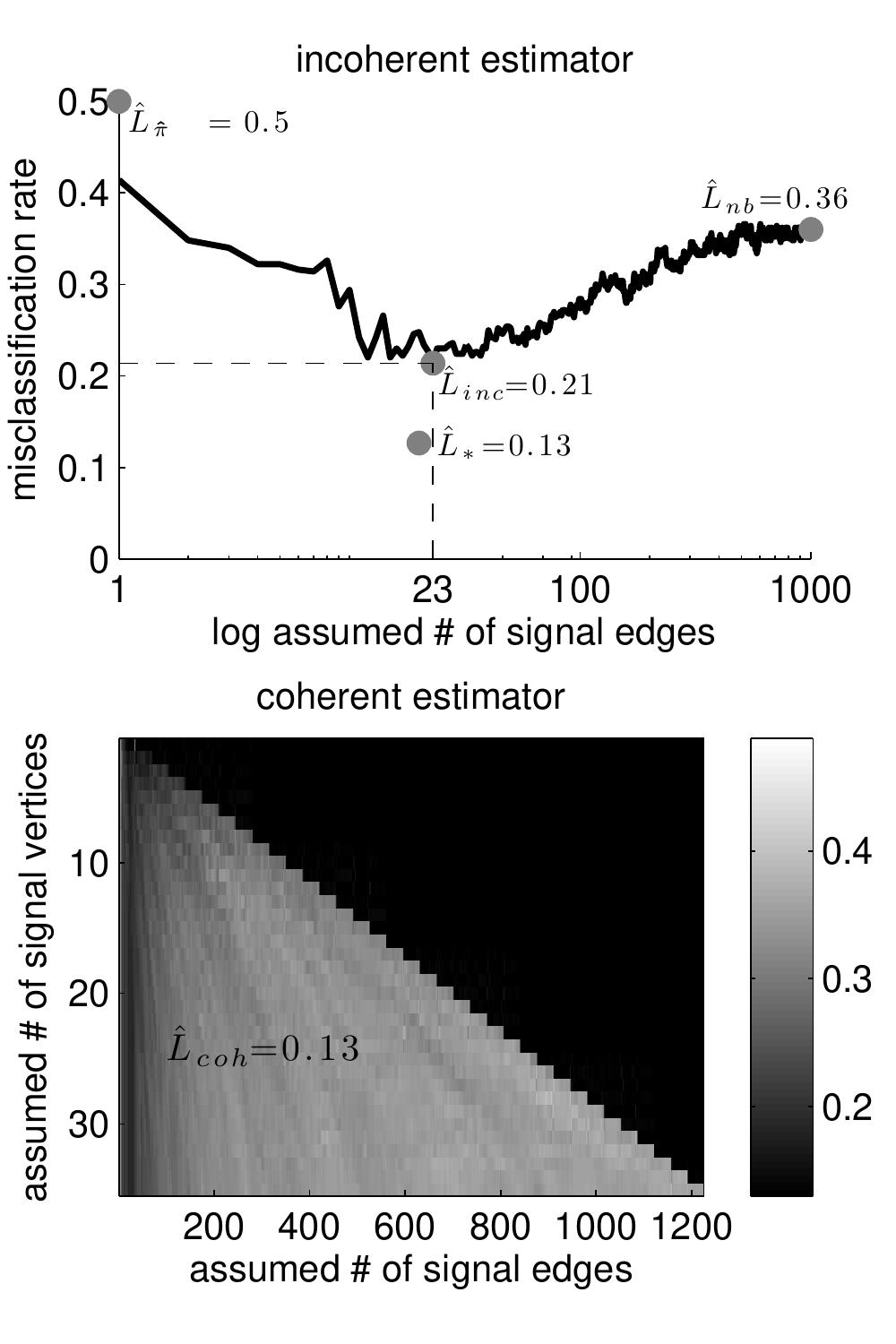}
	\caption{When constraints on the number of signal-edges ($s$) or signal-vertices ($m$) are unknown, a search over these hyperparameters can yield estimates $\wh{s}$ and $\wh{m}$.  Both panels depict held-out cross-validation error as a function of varying these parameters for the model 
	$\mc{M}_{70}(1,20;0.5,0.1,0.3)$
	(the same as in Figures \ref{fig:4x4} and \ref{fig:homo}), using 200 training samples and 500 test samples, with $m=1$ and $s=20$.  The top panel depicts misclassification rate of the incoherent estimator as a function of the number of estimated signal-edges on a log scale, with the best performing classifier achieving $\wh{L}_{inc}=0.21$. Note that in this simulation,  $s=20< \wh{s}_{inc}=23$.  This ``conservatism'' is typical and appropriate in many model selection situations.  The bottom panel shows $\wh{L}_{coh}$ as a function of both $m'$ and $s'$.  For this simulation, $\wh{m}_{coh}=1$ and $\wh{s}_{coh}=24$, further corroborating the conservative stance on model selection. Note that $L_{\mb{\pi}} > \wh{L}_{nb} > \wh{L}_{inc} > \wh{L}_{coh} \geq L_*$ as one would hope for this coherent simulation.  Incidentally, the coherent classifier achieved Bayes error here, $L_*=0.13$.}
	\label{fig:coherent}
\end{figure}


\section{MR Connectome \tk{Sex} Classification} 
\label{sub:sex}

A connectome is brain-graph \cite{Sporns2010}.  MR connectomes utilize multi-modal Magnetic Resonance (MR) imaging to determine both the vertex and edge set for each individual \cite{Hagmann2010}.  This section investigates the utility of the classifiers developed above on data collected for the Baltimore Longitudinal Study of Aging, as described previously \cite{OHBM10}.  Briefly, 49 subjects (25 male, 24 female) underwent a diffusion-weighted MRI protocol. The Magnetic Resonance Connectome Automated Pipeline (MRCAP) was used to convert each subject's raw multi-modal MR data into a connectome \cite{MRCAP11} (each connectome is a simple graph with 70 vertices and up to $\binom{70}{2}=2415$ edges). Lacking strong priors on either the number of signal-edges or signal-vertices in the signal-subgraph (or even whether a signal-subgraph exists), we searched over a large space of hyper-parameters using \tk{leave-one-out} cross-validated misclassification performance as our metric of success (Figure \ref{fig:data}).  The na\"ive Bayes classifier---which assumes the signal-subgraph is the whole edge set, $\mhc{S}_{nb}=\mc{E}$---performs marginally better than chance: $\wh{L}_{nb}=0.41$ (p-value $\approx 0.05$ assessed by a permutation test).  With a relatively small number of incoherent edges---$\wh{s}_{inc}=10$---the incoherent classifier (top left panel) achieves $\wh{L}_{inc}=0.27$, significantly better than chance (p-value $<0.0007$), but not significantly better than the na\"ive Bayes classifier (using McNemar's test).  The coherent classifier achieved a minimum of $\wh{L}_{coh}=0.16$ (top right and middle panels), 
\comment{not significantly better than the incoherent classifier, but }
significantly better than both chance and the na\"ive Bayes classifier (p-values $<10^{-5}$ and $<0.004$, respectively).  This improved performance upon using the coherent classifier suggests that the signal-subgraph is at least approximately coherent. Using $\wh{m}_{coh}=12$ and $\wh{s}_{coh}=360$ from the best performing coherent classifier, we can estimate the signal-subgraph (bottom left).  The coherogram suggests that indeed, the signal is somewhat, but not entirely coherent (bottom right).

\tk{We next compare the performance of our classifiers on this MR connectome sex classification data set to several other classifiers. First, a standard parametric classifier: lasso.  We chose the regularization parameter via a 10-fold cross-validation.  Second, a non-parametric (distribution free) classifier: $k_n$-nearest neighbor ($k$NN), which operates directly on graphs \cite{VP11_super}.  This $k$NN classifier uses the Frobenius norm distance metric.  We tried all $k \in [n]$ and simply report the best performance. The universal consistency of this $k$NN classifier is useful in assessing the algorithm complexity supported by this data.  In particular, given enough samples, $k$NN will achieve optimal performance.  Less than optimal performance therefore indicates that the sample size is not sufficiently large for this $k$NN classifier.  Third, a graph invariant based classifier.  We computed six graph invariants for each graph: size, max degree, scan statistic, number of triangles, clustering coefficient, and average path length, normalized each to have zero mean and unit variance, and then used a $k$NN with $\ell_2$ distance metric on the invariants. These particular invariants were chosen based on their desirable statistical properties \cite{PCP10, priebe2010you, Rukhin2011}.} 

\tk{Despite the small sample size, Table \ref{tab:bakeoff} demonstrates that the signal-subgraph classifier is significantly better than all the others, as assessed via a one-sided McNemar's test.}  

\begin{table}
	\caption{\tk{Bake-off comparing a number of different classifiers on the MR connectome sex classification data.  Error indicates misclassification error using the best hyper-parameters found for each classifier. P-value indicates the p-value of a one-sided McNemar's test comparing each classifier to the best signal-subgraph classifier.  The signal-subgraph classifier is significantly better than all the others.}  }
	\begin{center}
\begin{tabular}{|c|c|c|}
	\hline
classifier & error & p-val \\
\hline \hline
prior 			& 0.50 & $<0.01$ \\ \hline
na\"ive Bayes 	& 0.41 & $<0.01$ 		\\ \hline
lasso 			& 0.27 & $<0.02$ 				  	\\ \hline
graph-$k$NN 	& 0.35 & $<0.02$ 	\\ \hline
invariant-$k$NN & 0.43 & $<0.01$		\\ \hline
signal-subgraph 		& \textbf{0.16} & n$\backslash$a \\ \hline
\end{tabular}
\end{center}
\label{tab:bakeoff}
\end{table}

\begin{figure}[htbp]
	\centering
		\includegraphics[width=1.0\linewidth]{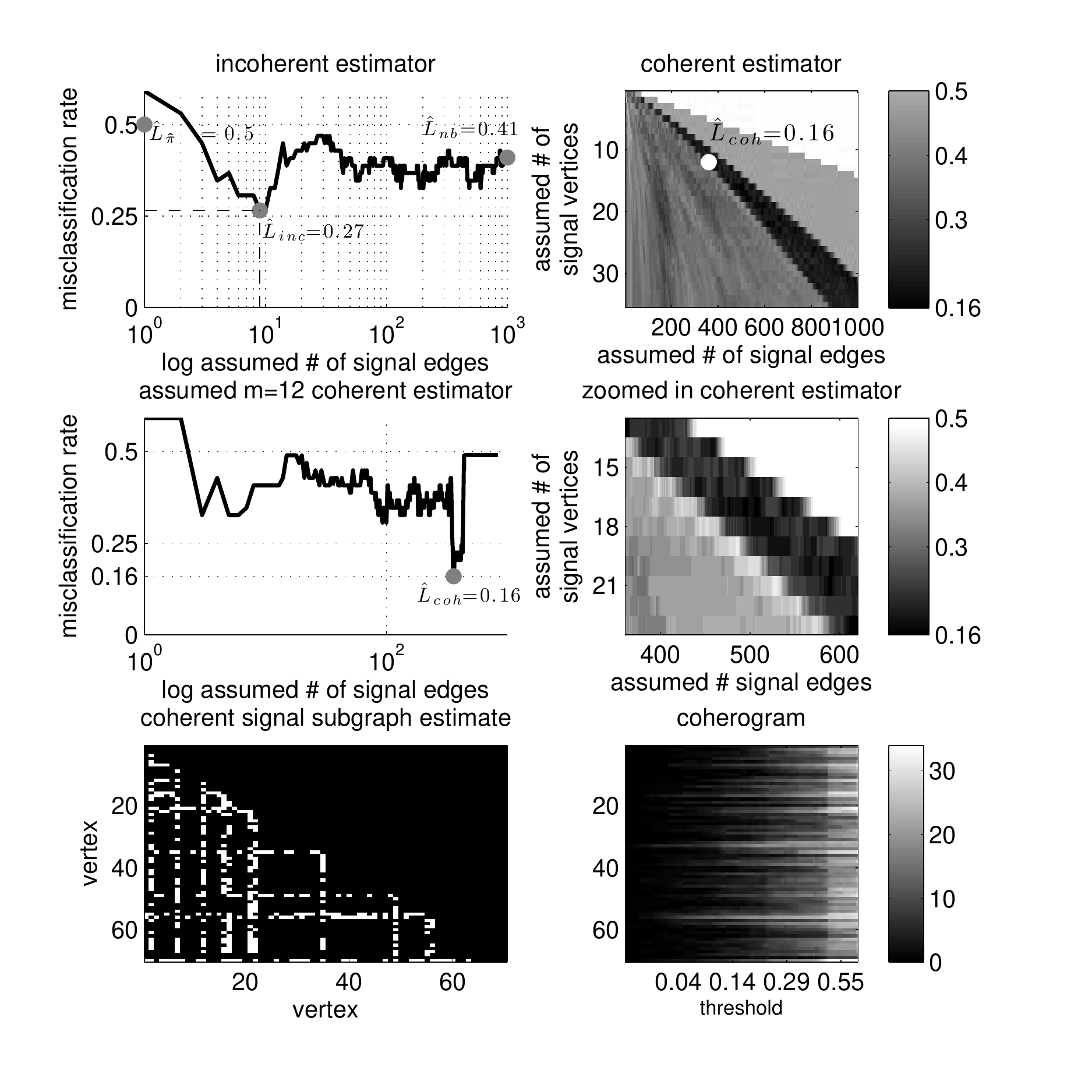}
	\caption{MR connectome sex signal-subgraph estimation and analysis. By cross-validating over hyperparameters and models, we estimate that the ``best'' incoherent signal-subgraph (for this inference task on these data) has $\wh{s}_{inc}=10$ and yields a misclassification rate of $\wh{L}_{inc}=0.27$, whereas the best coherent signal-subgraph has $\wh{m}_{coh}=12$ and $\wh{s}_{coh}=360$, achieving $\wh{L}_{coh}=0.16$.  The top two panels depict the same information as Figure 5.  The middle two depict misclassification rate (left) for different choices of $m'=12$ as a function of $s'$ and (right) a zoomed-in depiction of the top right panel. The bottom left panel shows the estimated signal-subgraph, and the bottom right shows the coherogram.  Together, these bottom panels suggest that the signal-subgraph for these data is at least somewhat coherent.}
	\label{fig:data}
\end{figure}


\subsection{Model Evaluation} 
\label{sub:model_checking}

\comment{Although the signal-subgraph classifier employed above estimated a signal-subgraph, w}We investigate to what extent the \tk{above} estimated signal-subgraph represents the true signal-subgraph.  We address this question in two ways:  (i) synthetic data analysis and (ii) assumption checking.  

\subsubsection{Synthetic Data Analysis} 
\label{ssub:synthetic_data_analysis}

For the synthetic data analysis, we generated data as follows.  Given the above estimated signal-subgraph, for every edge not in $\mhc{S}_n$, let $p_{uv|0}=p_{uv|1}=\wh{p}_{uv}$, where $\wh{p}_{uv}$ is the estimated edge probability averaging over all samples.  For all edges in $\mhc{S}_n$, let $p_{uv|y}=\wh{p}_{uv|y}$.  Set the priors according to the data as well: $\pi=\wh{\pi}$.  

\begin{figure*}[tb!]
	\centering
		\includegraphics[width=0.7\linewidth]{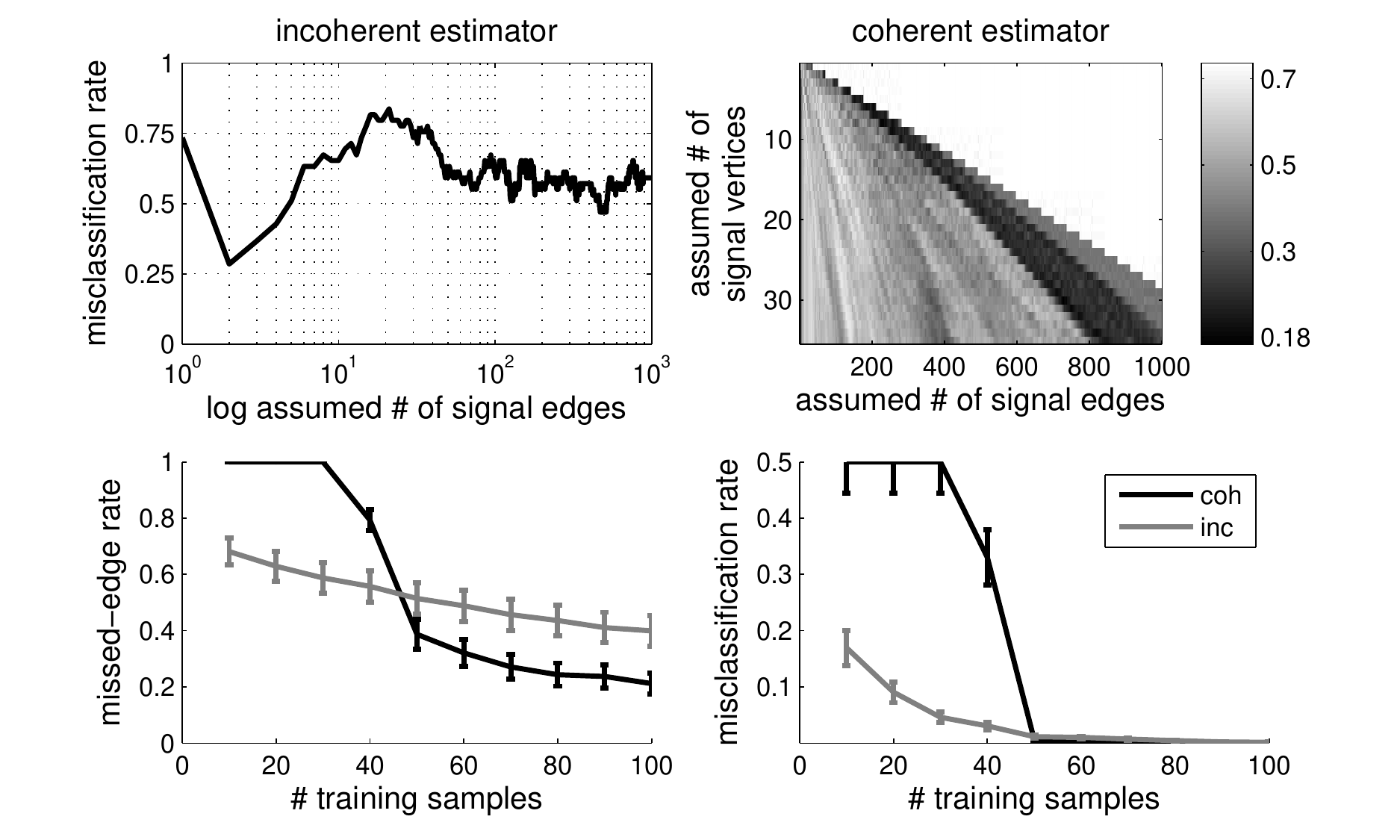}
	\caption{Synthetic data analysis provides some intuition for model checking and future improvements.  The top two panels show the incoherent (left) and coherent (right) misclassification rates as a function of the hyper-parameter choices for $n=49$.  These plots look quite similar to those obtained in the real connectome data (Figure \ref{fig:data}), which suggests that the chosen model may be adequate.  The bottom panels show the missed-edge rate (left) and misclassification rate (right) as a function of the number of training samples.  With about 50 training samples, approximately half of the edges identified by each classifier are true edges.  Additionally, slightly more than 50 training samples seems to be sufficient for obtaining nearly perfect classification, suggesting that perhaps only a few more subjects would be sufficient to yield much greater classification performance.}
	\label{fig:synthetic}
\end{figure*}

Given this synthetic data model, we first generated 49 data samples, 25 from class 0 and 24 from class 1, and estimated the incoherent and coherent classifier performance on a single synthetic experiment (Figure \ref{fig:synthetic}, top panels).  The performance of the classifiers on the synthetic data qualitatively mirrors that of the real data, suggesting some degree of model appropriateness.  To assess what fraction of the edges in the estimated signal-subgraph were reliable, even assuming a true model, we then sampled up to 100 training samples (and 100 test samples), and computed the missed-edge rate (bottom left) and misclassification rate (bottom right) as a function of the number of samples.  Given approximately 50 samples, the incoherent signal-subgraph estimator correctly identifies about $40\%$ of the edges, whereas the coherent signal-subgraph estimator correctly identifies about $50\%$.  This suggests that even if the model were true (which we doubt) we are justified to believe that only about half the edges in the estimated signal-subgraph are in the actual signal-subgraph.  
\tk{Despite  our stated desideratum of interpretability of the resulting classifier in terms of correctly identifying the signal-edges and vertices, for data sampled from this assumed distribution, sample sizes of $<50$ seem to be insufficient.  That said, }
\comment{Moreover,} both missed-edge rate and misclassification rate exhibit a step-like function in performance: after about 50 samples, performance dramatically improves.  This suggests that perhaps only a few more data points would be necessary to obtain greatly improved classification accuracy.

\subsubsection{Model Checking} 
\label{ssub:model_checking}


The assumption of independence between edges is (i) very useful for algorithms and analysis, and (ii) almost certainly nonsense for real connectome data.  Checking whether edges are independent is relatively easy.  Figure \ref{fig:cov} shows the correlation coefficient between all pairs of edges in the estimated signal-subgraph from the neurobiological data.  We used a spectral clustering algorithm \cite{Dhillon2001} to  more clearly highlight any significant correlations.  Several groups of edges seem to be highly correlated.  To assess significance, we compare the distribution of correlation coefficients with the distribution of correlation coefficients obtained from the synthetic data analysis.  A two-sample Kolmogorov-Smirnov test shows that the two matrices are significantly different (p-value $\approx 0$), rejecting the null hypothesis that the edges in the real data are independent. This analysis further corroborates that making independence assumptions can be fruitful even when the data are dependent \cite{Hand2001}.

%
%
%
%
%

\begin{figure}[htbp]
	\centering
		\includegraphics[width=1.0\linewidth]{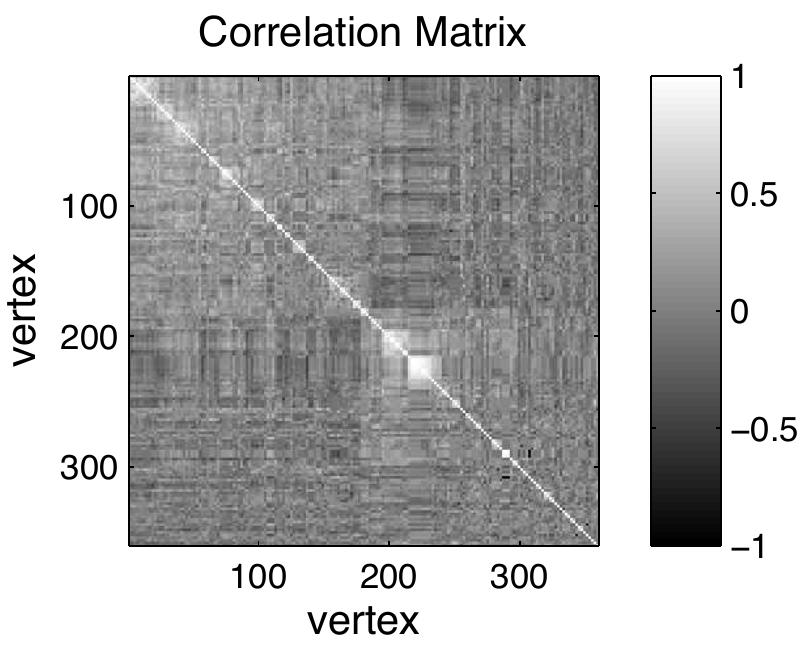}
	\caption{The correlation matrix between all the edges in the coherent signal-subgraph estimate. Edges are organized by co-clustering to highlight any similarities.  Although most edges are uncorrelated, several groups of edges cluster, indicative of the fact that the edges are not independent (p-value of $\approx 0$ using a two-sample Kolmogorov-Smirnov test comparing the real and synthetic correlation matrices). }
	\label{fig:cov}
\end{figure}





\section{Discussion} 
\label{sec:discussion}

This work makes the following contributions. First, it introduces a novel graph/class model that admits rigorous statistical investigation.  Moreover, it presents two approaches for estimating the signal-subgraph: the first using only vertex label information, the second also utilizing graph structure.  The resulting estimators have desirable asymptotic and finite sample properties, including consistency and robustness to various model misspecifications.  Third, simulated data analysis indicate that neither approach dominates the other; rather, the best approach is a function of both the model and the amount of training data. \tk{And while the lasso classifier has similar error properties to our incoherent classifier, lasso's computational time is about an order of magnitude longer.}
Fourth, these classifiers are applied to \tk{an MR connectome sex classification} \comment{a connectome} data set; the coherent classifier performs significantly better than \comment{both na\"ive Bayes and incoherent classifiers} \tk{a variety of benchmark classifiers}.  Fifth, synthetic data analysis suggests that while we can use the signal-subgraph estimators to improve classification performance, we should not expect that all the edges in the estimated signal-subgraph will be the true signal-edges, even when the model is correct. Moreover, we might expect a drastic improvement in classification performance with only a few additional data samples.  Finally, model checking suggests that the independent edge assumption does not fit the data well.  

\tk{Our signal-subgraph classifiers represent somewhat of a departure from previous work.  Most graph classification algorithms come from the ``structural pattern recognition'' school of thought \cite{Bunke2011}, lacking an explicit statistical model and associated provable properties. On the other hand, most work on ``statistical pattern recognition'' begins by assuming the data to be classified are Euclidean vectors \cite{Devroye1997}.  Our work is a unification of the two.  Moreover, because the sufficient statistics are essentially encoded in a matrix, our work can be related to recent developments in matrix decompositions.  For example, sparse and low-rank matrix decompositions are close in spirit to our coherent signal subgraph estimators \cite{Candes2009b, Ding2011, Chandrasekaran2011}. Note, however, that our coherent estimator is robust to signal-vertices having a subset of its edges highly non-significant; that is, the coherent signal-subgraph estimator can be thought of as a \emph{local} sparse and low-rank decomposition.}

Collectively, the above analyses suggest a number of possible next steps.  First, collect more data.  Second, relax various assumptions, including (i) the independent edge assumption by considering conditionally independent edges \cite{Hoff02,STFP,Fishkind2012}, (ii) binary edge and class assumptions, and (iii) labeled vertices assumption. 
\tk{Specifically, extension to situations for which none of the vertices are labeled \cite{VP11_QAP, VP11_unlabeled},  only some subset of vertices are labeled \cite{VN,BVN}, or data are otherwise errorfully observed \cite{Bock2011}, are all avenues of future investigation.}
Third, transform a number of conjectures that have arisen due to these results into theorems.  For instance, perhaps the misclassification rate is a monotonic function of the missed-edge rate.  Fourth, (Bayesian) model-averaging to combine estimated signal-subgraphs instead of picking one might improve performance (perhaps at the cost of computational resources and interpretability).  

We hope the proposed approaches will yield many applications.  To that end, all the data and code used in this work is available from the author's website, \url{http://jovo.me}.  

%
%
%

%

\ifCLASSOPTIONcompsoc
  \section*{Acknowledgments}
\else
  \section*{Acknowledgment}
\fi

This work was partially supported by the Research Program in Applied Neuroscience. The authors would like to thank Michael Trosset for a helpful suggestion.

\ifCLASSOPTIONcaptionsoff
  \newpage
\fi

\bibliography{library}
\bibliographystyle{IEEEtran}

\begin{IEEEbiographynophoto}{Joshua T. Vogelstein}
Joshua T. Vogelstein received a B.S degree from the Department of Biomedical Engineering at Washington University in St. Louis, MO in 2002, a M.S. degree from the Department of Applied Mathematics \& Statistics at Johns Hopkins University (JHU) in Baltimore, MD in 2009, and a Ph.D. degree from the Department of Neuroscience at Johns Hopkins School of Medicine in Baltimore, MD in 2009.  He is currently an Assistant Research Scientist in the Department of Applied Mathematics and Statistics at JHU, with a joint appointment in the Human Language Technology Center of Excellence.  His research interests primarily include statistical connectomics, including theory and applications for high-dimensional graph-valued data. His research has been featured in a number of prominent scientific and engineering journals including Annals of Applied Statistics, IEEE Transactions on Neural Systems and Rehabilitation Engineering, Nature Neuroscience, and Science Translational Medicine.
\end{IEEEbiographynophoto}

\begin{IEEEbiographynophoto}{William R. Gray}
William R. Gray graduated from Vanderbilt University in 2003 with a Bachelor’s degree in electrical engineering, and received his MS in electrical engineering in 2005 from the University of Southern California.  Currently, Will is a PhD student in electrical engineering at Johns Hopkins University, where he is conducting research in the areas of connectivity, signal and image processing, and machine learning.  He is also a member of the technical staff at the Johns Hopkins University Applied Physics Laboratory, where he manages projects in the Biomedicine and Undersea Warfare business areas.  Will is a member of IEEE, Eta Kappa Nu, and Tau Beta Pi
\end{IEEEbiographynophoto}


\begin{IEEEbiographynophoto}{R. Jacob Vogelstein}
R. Jacob Vogelstein received the Sc.B. degree in neuroengineering from Brown University, Providence, RI, and the Ph.D. degree in biomedical engineering from the Johns Hopkins University School of Medicine, Baltimore, MD.  He currently oversees the Applied Neuroscience programs at the Johns Hopkins University Applied Physics Laboratory as an Assistant Program Manager, and has an appointment as an Assistant Research Professor at the JHU Whiting School of Engineering’s Department of Electrical and Computer Engineering. He has worked on neuroscience technology for over a decade, focusing primarily on neuromorphic systems and closed-loop brain–machine interfaces. His research has been featured in a number of prominent scientific and engineering journals including the IEEE Transactions on Neural Systems and Rehabilitation Engineering, the IEEE Transactions on Biomedical Circuits and Systems, and the IEEE Transactions on Neural Networks.  
\end{IEEEbiographynophoto}

\begin{IEEEbiographynophoto}{Carey E. Priebe}
Carey E. Priebe received the B.S. degree in mathematics from Purdue University in 1984, the M.S. degree in computer science from San Diego State University in 1988, and the Ph.D. degree in information technology (computational statistics) from George Mason University in 1993. From 1985 to 1994 he worked as a mathematician and scientist in the US Navy research and development laboratory system. Since 1994 he has been a professor in the Department of Applied Mathematics and Statistics, Whiting School of Engineering, Johns Hopkins University, Baltimore, Maryland. At Johns Hopkins, he holds joint appointments in the Department of Computer Science, the Department of Electrical and Computer Engineering, the Center for Imaging Science, the Human Language Technology Center of Excellence, and the Whitaker Biomedical Engineering Institute. He is a past President of the Interface Foundation of North America - Computing Science \& Statistics, a past Chair of the American Statistical Association Section on Statistical Computing, a past Vice President of the International Association for Statistical Computing, and on the editorial boards of Journal of Computational and Graphical Statistics, Computational Statistics and Data Analysis, and Computational Statistics. His research interests include computational statistics, kernel and mixture estimates, statistical pattern recognition, statistical image analysis, dimensionality reduction, model selection, and statistical inference for high-dimensional and graph data. He is a Senior Member of the IEEE, a Lifetime Member of the Institute of Mathematical Statistics, an Elected Member of the International Statistical Institute, and a Fellow of the American Statistical Association.
\end{IEEEbiographynophoto}



\end{document}